\documentclass[conference]{IEEEtran}
\usepackage{cite}      % Written by Donald Arseneau
\usepackage[usenames,dvipsnames,svgnames,table]{xcolor}
\usepackage{graphicx}  % Written by David Carlisle and Sebastian Rahtz
\usepackage{algorithm}
\usepackage{algpseudocode}

\usepackage{psfrag}    % Written by Craig Barratt, Michael C. Grant,
\usepackage{url}       % Written by Donald Arseneau
\usepackage{tikz}
\usepackage{makecell}

\usepackage{amsmath}   % From the American Mathematical Society
\usepackage{amssymb}

\newcommand{\shorten}[1]{}
\newtheorem{proposition}{Proposition}

\newtheorem{definition}{Definition}

\newtheorem{lemma}{Lemma}

\newtheorem{example}{Example}

\newcommand{\signed}%
    {{\unskip\nobreak\hfill\penalty50
      \hskip2em\hbox{}\nobreak\hfil $\blacksquare$
      \parfillskip=0pt \finalhyphendemerits=0 \par}}
\newenvironment{proof}[1]
    {
    \bf{Proof:}\rm{\noindent{#1 }}\ignorespaces
    }
    {\signed\addvspace\medskipamount}

\begin{document}

\title{Balanced Locally Repairable Codes}

% avoiding spaces at the end of the author lines is not a problem with
% conference papers because we don't use \thanks or \IEEEmembership
\author{\authorblockN{
Katina Kralevska, Danilo Gligoroski and Harald {\O}verby}
\authorblockA{Department of Telematics, Faculty of Information
Technology, Mathematics and Electrical Engineering, \\ NTNU, Norwegian University of Science and Technology, Trondheim, Norway,\\ Email:
\{katinak, danilog, haraldov\}@item.ntnu.no}

%\and
%\IEEEauthorblockN{James Kirk\\ and Montgomery Scott}
%\IEEEauthorblockA{Starfleet Academy\\
%San Francisco, California 96678-2391\\
}

%\specialpapernotice{(Invited Paper)}

\maketitle

\begin{abstract}
We introduce a family of balanced locally repairable codes (BLRCs) $[n, k, d]$ for arbitrary values of $n$, $k$ and $d$. Similar to other locally repairable codes (LRCs), the presented codes are suitable for applications that require a low repair locality. The novelty that we introduce in our construction is that we relax the strict requirement the repair locality to be a fixed small number $l$, and we allow the repair locality to be either $l$ or $l+1$.
This gives us the flexibility to construct BLRCs for arbitrary values of $n$ and $k$ which partially solves the open problem of finding a general construction of LRCs. Additionally, the relaxed locality criteria gives us an opportunity to search for BLRCs that have a low repair locality even when double failures occur. We use metrics such as a storage overhead, an average repair bandwidth, a Mean Time To Data Loss (MTTDL) and an update complexity to compare the performance of BLRCs with existing LRCs.
\end{abstract}

{\bfseries {Keywords}}: Locally Repairable codes, Balanced, Storage overhead, Update complexity, Repair bandwidth, MTTDL

% for peer review papers, inserts a page break and creates the second title.
% Will be ignored for other modes.
\IEEEpeerreviewmaketitle

\section{Introduction}

A conventional approach for achieving reliability in big data distributed storage systems is replication. In  particular, the reliability of 3-replication is an accepted industry standard for management of hardware failures and recovery. That is an apparent situation in systems such as Hadoop HDFS \cite{borthakur2007hadoop}, OpenStack SWIFT \cite{arnold2014openstack} or Microsoft Azure \cite{calder2011windows}. However, the accelerated and relentless data growth has made erasure coding a valuable alternative to 3-replication since erasure coding provides the same reliability as 3-replication, but with significant less storage overhead. Recently, there have been several proposals and experimental beta implementations of different types of erasure codes for huge distributed storage systems \cite{Fan:2009:DRD:1713072.1713075,plank2008jerasure,DBLP:journals/corr/KralevskaGJO16}.

Besides the reliability and the storage overhead, another important feature in distributed storage systems is the efficiency of the repair of a failed node. The efficiency is measured with two metrics: the repair bandwidth and the repair locality. The repair bandwidth is the amount of transferred data during a node repair, while the repair locality is the number of nodes contacted during the node repair process. Two types of erasure codes that address the repair efficiency have emerged: Regenerating codes \cite{5550492} and Locally Repairable Codes (LRCs) \cite{journals/tit/GopalanHSY12,conf/infocom/OggierD11,6195703}.

Regenerating codes \cite{5550492} aim to minimize the repair bandwidth, while LRCs seek to minimize the repair locality.
The main idea behind regenerating codes is using a sub-packetization.
Each block is divided into sub-packets and a recovery is performed by transferring sub-packets from all $n-1$ non-failed nodes that results in high I/O.
A proposal for reducing the I/O is given in \cite{conf/fast/KhanBPPH12}.
On the other hand, LRCs address the issue of accessing less nodes, but the amount of transferred data is bigger compared to regenerating codes.
However, communicating less nodes is beneficial for storage applications that require low I/O.

An $[n, k, d]_q$ MDS code $C$ has to transfer $k$ symbols to recover one lost symbol.
LRCs were independently introduced in \cite{journals/tit/GopalanHSY12,conf/infocom/OggierD11,6195703}.
The code $C$ has a locality $l$ if the $i-$th code symbol $c_i$, $1\leq i \leq n$, can be recovered by accessing $l$ symbols where $l < k$.
It was proved in \cite{journals/tit/GopalanHSY12} that the minimum distance of an $[n, k, d]_q$ code with a locality $l$ is
\begin{equation}
d \leq n-k+2-\Bigl\lceil \frac{k}{l} \Bigr\rceil.
\label{distance}
\end{equation}
Huang et al. showed the existence of pyramid codes that achieve this distance when the field size is big enough \cite{4276609}.
\newline
Two practical LRCs have been implemented in Windows Azure Storage \cite{conf/usenix/HuangSXOCG0Y12} and HDFS-Xorbas by Facebook \cite{journals/pvldb/SathiamoorthyAPDVCB13}.
%The concept of locality is used in Microsoft LRC \cite{conf/usenix/HuangSXOCG0Y12} and Facebook Xorbas \cite{journals/pvldb/SathiamoorthyAPDVCB13} codes.
Both implementations reduce the repair bandwidth and the I/O for reconstructing a single data block by introducing a fixed number of $l$ local and $r$ global parity blocks.
Any single data block can be recovered from $\frac{k}{l}$ blocks within its local group.
%LRC decode $r+1$ arbitrary block failures and up to $l+r$ theoretically decodable failures.
However, reconstruction of any global parity block (in Windows Azure) or double blocks failures is performed in the same way as with Reed-Solomon (RS) codes, i.e., $k$ blocks need to be transferred.
%An implied parity block is used in HDFS-Xorbas to have an efficient repair of any single global parity block.

Since node failures in storage systems are often correlated \cite{conf/osdi/FordLPSTBGQ10}, there is a need for other erasure codes than LRCs for recovery from multiple failures.
For instance, Shingled erasure codes (SHEC) have a low average repair bandwidth when multiple failures occur, but they are not so efficient in terms of the storage overhead and the reliability \cite{conf/hotdep/MiyamaeNS14}.

%The failures in real storage systems are often correlated.
%Shingled erasure codes (SHEC) are codes with local parity groups shingled with each other that provide efficient recovery from multiple failures \cite{conf/hotdep/MiyamaeNS14}.
%Other codes with locality that are efficient against multiple failures are presented in \cite{6882150,oai:arXiv.org:1501.06683}.
%As we said LRC usually use the global parities for multiple disk failures.

The locality also has an impact on the update complexity \cite{5706931}.
This is particularly important for hot data, i.e., frequently accessed data.
For instance, an $[16, 10]$ LRC where the locality is 5, writing a data block takes 6 write operations (1 write to itself, 1 write to the local parity and 4 writes to the global parities).

Thus, having a general construction of LRCs that are simultaneously optimal in terms of storage overhead, reliability, locality and update complexity for a single failure and double failures is an important problem that is addressed in this work.

%\cite{959255}
%\cite{conf/fast/KhanBPPH12}

%scheme, writing a data block takes 5 write operations (1 write to itself, 4 reads of the redundant blocks, 4 operations to compute the change, and then 4 writes to the redundant blocks [1]), and reading takes 12 read operations when hitting a failed data block (in order to perform the inverse of the mathematical transform).

\subsection{Our Contribution}
We define a new family of balanced locally repairable codes (BLRCs).
One of their main characteristics is that every systematic block has an equal (balanced) influence to the parity blocks.
That is to say, each systematic block affects exactly $w$ parity blocks.
Additionally, we pay attention on the repair locality.
In our construction we use a similar (but yet different) approach to the approaches introduced by Luby et al. for the construction of irregular LDPC codes \cite{luby2001improved} and Garcia-Frias and Zhong for the construction of regular and irregular LDGM codes \cite{garcia2003approaching}. Namely, instead of the strict requirement the repair locality to be a fixed small number $l$, it may be either $l$ or $l+1$.
This partially solves the open problem given by Tamo et al. about a general construction of LRCs \cite{6620540}, because we construct LRCs for arbitrary values of $n$ and $k$, but the locality is not strictly equal to $l$.
%which solves the
Moreover, the relaxed locality criteria gives us an opportunity to search for BLRCs that have a low repair locality even when double failures occur.
%The focus of the present paper is on codes with locality.
%We refer to these codes as balanced locally repairable codes.
We use four metrics to examine the performance of BLRCs:
\begin{itemize}
%\vspace{-0.15cm}
  \item Storage overhead (a ratio of the parity to the data blocks $\frac{r}{k}$);
%\vspace{-0.15cm}
  \item Average repair bandwidth (a ratio of the repair bandwidth to repair both data and parity blocks to the total stored data (sum of the data and the parity blocks));
%\vspace{-0.15cm}
  \item MTTDL (Mean Time To Data Loss - an estimate of the expected time that it would take a given storage system to exhibit enough failures such that at least one block of data cannot be retrieved or reconstructed);
%\vspace{-0.15cm}
  \item Update complexity (a maximum number of elements that must be updated when any single element is changed).
\end{itemize}
%\vspace{-0.15cm}

In summary, several goals are achieved simultaneously with this work: 1) low storage overhead; 2) low average repair bandwidth for single and double failures; 3) high reliability; and 4) improved update performance.
%simplified design, implementation and verification;
%Data storage applications require codes with small storage overhead, low locality and large distance.

The paper is organized as follows.
In Section \ref{Preliminaries}, we introduce the terminology and the definition of balanced locally repairable codes.
In Section \ref{Example}, we give code examples and examine their performance by using the predefined metrics.
We also compare the properties of our codes to 3-replication, RS and other LRCs.
A reliability analysis is presented in Section \ref{reliability}.
Conclusions are summarized in Section \ref{Conclusions}.

\section{Definition of Balanced Locally Repairable codes} \label{Preliminaries}
%\vspace{-0.3cm}
We use the following notations throughout the rest of the paper.
A file of size $M$ is divided into $k$ equally sized blocks and encoded in $GF(q)$ with an $[n, k, d]_q$ code into $n$ coded blocks. %Each block has size $\frac{M}{k}$.
An $[n, k, d]_q$ code is called maximum distance separable (MDS) if $d = n - k + 1$.
An $[n, k, d]_q$ MDS code reconstructs a failed block from any $k$ out of the $n$ blocks.
%The Singleton defect of an $[n, k, d]_q$ code $C$ defined as $s(C ) = n - k + 1 - d$ measures how far away is $C$ from being a MDS code.
We denote the number of parity blocks with $r = n - k$.

%\begin{definition}[\hspace{-0.025cm}\cite{DCC::Boer1996}]
%An $[n, k, d]_q$ code $C$ with Singleton defect $s(C)=1$ is called an almost MDS (AMDS) code.
%\label{AlmostMDS}
%\end{definition}

\begin{definition}\label{def:DeErasure}
Let $C$ be an $[n, k, d]_q$ code over $GF(q)$ with a generator matrix $G$:
\begin{equation}
G=
\begin{bmatrix}
I_k |
P
\end{bmatrix},
\end{equation}
where $I_k$ is an identity matrix of order $k$ and the $k \times (n-k)$ matrix $P$ specifies how the parity is defined for the given $[n, k, d]_q$ linear code.
%We will call the matrix $P$ \emph{the parity part} (not to be confused with the Parity-check matrix, although there is a clear corelation between these two terms).
We call $C$ a \emph{Balanced Locally Repairable Code (BLRC)}, if the Hamming weight of every row in the matrix $P$ is $w$ where $w<k$, the Hamming weight of every column is $l$ or $l+1$ and for every submatrix $P'$ of $P$ consisting of $v$ rows, $1\leq v \leq w$, from $P$ it holds that $Rank(P')=v$.
\label{Generator}
\end{definition}
The field size should be big enough so that the condition for the rank in Definition \ref{Generator} is fullfiled.
%\vspace{-0.2cm}
%$P$ is a sparse matrix.
\begin{example}
	Let us consider the following $[13,8,3]$ code with a generator matrix:
	$$
      G=\left[\small
      \begin{array}{c@{\hspace{0.5em}}c@{\hspace{0.5em}}c@{\hspace{0.5em}}c@{\hspace{0.5em}}c@{\hspace{0.5em}}c@{\hspace{0.5em}}c@{\hspace{0.5em}}c@{\hspace{0.5em}}c@{\hspace{0.5em}}c@{\hspace{0.5em}}c@{\hspace{0.5em}}c@{\hspace{0.5em}}c}
      1 & 0 & 0 & 0 & 0 & 0 & 0 & 0 & 0       & 0        & c_{1,11} & c_{1,12} & 0 \\
      0 & 1 & 0 & 0 & 0 & 0 & 0 & 0 & c_{2,9} & 0        & 0        & c_{2,12} & 0 \\
      0 & 0 & 1 & 0 & 0 & 0 & 0 & 0 & c_{3,9} & 0        & c_{3,11} &   0        & 0 \\
      0 & 0 & 0 & 1 & 0 & 0 & 0 & 0 & 0       & 0        & c_{4,11} &  0        & c_{4,13} \\
      0 & 0 & 0 & 0 & 1 & 0 & 0 & 0 & 0       & 0        & 0        & c_{5,12} & c_{5,13} \\
      0 & 0 & 0 & 0 & 0 & 1 & 0 & 0 & 0       & c_{6,10} & 0        &  0        & c_{6,13} \\
      0 & 0 & 0 & 0 & 0 & 0 & 1 & 0 & 0       & c_{7,10} & c_{7,11} &  0        & 0 \\
      0 & 0 & 0 & 0 & 0 & 0 & 0 & 1 & c_{8,9} & c_{8,10} & 0        &  0        & 0 \\
      \end{array}
      \right],
	$$
	where $c_{i,j}$ are some nonzero elements from $GF(q)$. Note that the Hamming weight of every row in $P$ is $w=2 < k$, while the Hamming weight of every column in $P$ is either $l=3$ or $l+1 = 4$. Finally, since any two rows in $P$ are linearly independent, the rank condition from Definition \ref{def:DeErasure} is fulfilled.
	Thus, the code is a balanced locally repairable code.
\end{example}

From the erasure recovery point of view, we use the minimum distance of the code as a metric for its fault tolerance. We have the following Lemma:
\begin{lemma} \label{MinDistanceDeErasure} If $[n, k, d]_q$ is a balanced locally repairable code defined in a finite field $GF(q)$, then
\begin{equation}
d = w+1.
\end{equation}
\end{lemma}
%\vspace{-0.2cm}
\begin{proof}
	The minimum distance $d$ of a code $C$ is equal to the number of failed blocks (erasures) after which the data cannot be recovered.
Note that if one systematic block and $w$ parity blocks that are linear combinations of the specific systematic block fail, then the systematic block is non-recoverable. This is true due to the fact that all $w+1$ parts that have (non-encoded or encoded) information about the systematic block have been lost. Thus, it follows that $d \le w+1$. 
Let us assume that $w_s$ systematic and $w_p$ parity blocks are lost where $w=w_s+w_p$. If we consider that the lost $w_p$ parity blocks are linear combinations from the $w_s$ systematic blocks that have been also lost, then the systematic blocks can be recovered only if for every submatrix $P'$ consisting of $w_s$ rows of $P$ it holds that $Rank(P')=w_s$. After recovering the systematic blocks, the lost parity blocks can be recovered.
Let us consider that $w_s=w$ systematic blocks have been lost. The lost systematic blocks can be recovered by selecting the corresponding rows that contain the specific $w$ systematic blocks in the matrix $P$ and producing a matrix $P'$. Since $Rank(P')=w$, the lost systematic blocks can be recovered. On the other hand, if $w_p=w$ parity blocks have been lost, then each of the parity blocks can be recovered from $l$ or $l+1$ systematic blocks. 
In any case $d>w$. Consequently, it follows that $d = w+1$.
\end{proof}

The locality of a systematic code $C$ is defined as the number of data blocks that each parity block is a function of.

\begin{lemma}
Let $[n, k, d]_q$ is a balanced locally repairable code defined in a finite field $GF(q)$. Then for its locality $l$, it holds: \begin{equation}
l = \Bigl\lfloor\frac{(d-1) \times k}{n-k}\Bigr\rfloor.\end{equation}
\end{lemma}
%\vspace{-0.2cm}
\begin{proof}
The parity part $P$ of the generator matrix $G$ is an $k \times (n-k)$ matrix. Since every row has $w$ nonzero elements, with $k$ such rows, the total number of nonzero elements in $P$ is $w \times k$.
It follows that the average number of nonzero elements in every column of $P$ is $l=\Bigl\lfloor\frac{(d-1) \times k}{n-k}\Bigr\rfloor$.
%	It follows from Lemma \ref{MinDistanceDeErasure} that for all $r=n-k$ parity blocks, there are $w=d-1$ nonzero elements in the parity part $P$. Then there are in total $w \times r = (d-1)\times (n-k)$ nonzero elements in the parity part $P$. From there, the average number of nonzero elements in every column of $P$ is $l = \lfloor\frac{(d-1) \times (n-k)}{k}\rfloor$.
\end{proof}

%For some use-case scenarios, the goal is the locality to be low due to the need for a low number of transferred blocks during a block repair process and low update complexity (the number of writes per update). The following two propositions describe these properties for the Locally Repairable Balanced codes.
The number of transferred blocks during a repair process and the update complexity for BLRCs are captured in the following propositions:

\begin{proposition}
When recovering one lost block (a systematic or parity block) in an $[n, k, d]_q$ balanced locally repairable code, the number of transferred blocks is $l$ or $l+1$.
\end{proposition}

The proof for Proposition 1 in connection with Lemma 2 includes a detailed algorithm how to construct BLRCs. We do not include it in this short paper due to space limitations, but we will include it in an extended version.

\begin{proposition}
The number of writes per update of an $[n, k, d]_q$ balanced locally repairable code is $w+1$.
\end{proposition}

Since the node failures in storage systems are often correlated \cite{conf/osdi/FordLPSTBGQ10}, we next give an algorithm for finding BLRCs that have a low repair locality even when two blocks have failed.
Algorithm 1 uses a stochastic hill-climbing search in a similar manner as in \cite{4658694,6875415}.
%\vspace{-0.15cm}
 %to find LRCs for arbitrary input parameters $n$, $k$ and $d$ that have low repair locality when two erasures are recovered.

%As we see, in the presented codes, the strict requirement for the repair locality to be a fixed small number $l$ is weakened and now it can be either $l$ or $l+1$. That gives us a flexibility to construct locally repairable codes for arbitrary values of $n$, $k$ and $d$. Additionally, the relaxed locality criteria gives us an opportunity to search for best LRCs that have low repair locality even when two blocks have failed. That is described in Algorithm 1 where in a similar manner as in \cite{6875415} we use a stochastic hill-climbing search. %to find LRCs for arbitrary input parameters $n$, $k$ and $d$ that have low repair locality when two erasures are recovered.

\begin{algorithm}
\caption{A general Stochastic Hill-Climbing search for finding a locally repairable code for given $n$, $k$ and $d$
\newline
\textbf{Input}: $n$, $k$ and $d$;
\newline
\textbf{Output}: A Balanced Locally Repairable Code.}
\label{Valid}
\begin{algorithmic}[1]
%\Procedure{%Access all elements from the $r^{j-1}$ rows of the non-failed systematic nodes where the elements from the $r^{j-1}$ are written}{}
%\For{$j=1$ to $k \setminus l$}
	\State Find a random $[n, k, d]$ linear code as in Definition \ref{Generator} where $w=d-1$ is the Hamming weight of every row of the matrix $P$;
	\State Repeatedly improve the solution by searching for codes with low average locality when two blocks failures have to be recovered, until no more improvements are necessary/possible.
\end{algorithmic}
\end{algorithm}
%\vspace{-0.15cm}
\begin{figure*}
\begin{minipage}[b]{0.5\linewidth}
\centering
\includegraphics[scale=0.35,height=3.7cm]{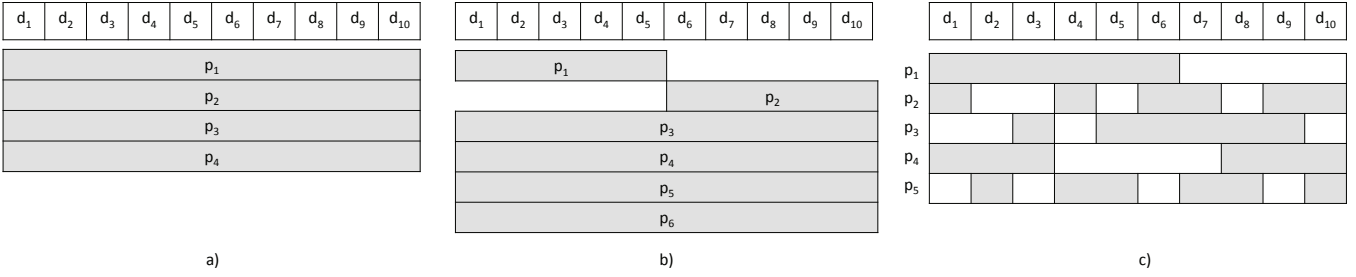}
\end{minipage}
\caption{(a) An [14, 10] RS code where $l=10$; (b) An [16, 10] LRC where $l=5$ for the local parities $p_1$ and $p_2$ and $l=10$ for the global parities $p_3, p_4, p_5$ and $p_6$; (c) An [15, 10] BLRC where $l=6$ and $w=3$}
\label{local}
\end{figure*}

The construction of our codes has some similarities with the construction of several classes of LDPC codes reported in the literature. In particular, several families of LDPC codes that are based on Finite Geometries are defined in \cite{6620540}. In that work, the restrictions that are inferred by the properties of Finite Geometries restrict the possible choices of different $n$ and $k$. Variable irregular LDPC codes are constructed by puncturing the codes or by splitting the columns and rows of the parity-check matrix $H$.
\newline
We have been inspired by two other works that are also from the area of LDPC codes. Namely, Luby et al. in 2001 introduced the principle of allowing irregularities for variable nodes in a LDPC construction \cite{luby2001improved}. They allowed those irregularities to have degree 2, 3, 4 or even 20, while in our construction the degree of locality is either $l$ or $l+1$. On the other hand, in 2003 Garcia-Frias and Zhong proposed regular and irregular LDGM codes in \cite{garcia2003approaching}. For the regular LDGM codes, the parity matrix $P$ has always a fixed row weight $X$ and fixed column weight $Y$ which is equivalent to the LRC case where the row weight is fixed at $w$ and the column weight is fixed at $l$. For the irregular LDGM codes the parity matrix $P$ has an average row weight $X$ and an average column weight $Y$, while in our construction the row weight is fixed to $w$ but the column weight can be either $l$ or $l+1$.

\section{Examples of Code Constructions} \label{Example}

In this Section we present several parity parts $P$ (not to be confused with a parity-check matrix $H$) of BLRCs for different code parameters.
%The locality defines the number of rows in the Latin rectangle and the number of columns is equal to the number of data blocks $k$.
%The generator matrix $G$ is obtained from the following $6 \times 10$ Latin rectangle:

The parity part $P_1$ of an $[15, 10]$ code for $l=6$ and $w=3$  is:

$$
      P_1=\left[\small
      \begin{array}{c@{\hspace{0.5em}}c@{\hspace{0.5em}}c@{\hspace{0.5em}}c@{\hspace{0.5em}}c@{\hspace{0.5em}}c@{\hspace{0.5em}}c@{\hspace{0.5em}}c@{\hspace{0.5em}}c@{\hspace{0.5em}}c@{\hspace{0.5em}}c@{\hspace{0.5em}}c@{\hspace{0.5em}}c}
 c_{1,11} & c_{1,12} & 0 & c_{1,14} & 0 \\
 c_{2,11} & 0 & 0 & c_{2,14} & c_{2,15} \\
 c_{3,11} & 0 & c_{3,13} & c_{3,14} & 0 \\
 c_{4,11} & c_{4,12} & 0 & 0 & c_{4,15} \\
 c_{5,11} & 0 & c_{5,13} & 0 & c_{5,15} \\
 c_{6,11} & c_{6,12} & c_{6,13} & 0 & 0 \\
 0 & c_{7,12} & c_{7,13} & 0 & c_{7,15} \\
 0 & 0 & c_{8,13} & c_{8,14} & c_{8,15} \\
0 & c_{9,12} & c_{9,13} & c_{9,14} & 0 \\
0 & c_{10,12} & 0 & c_{10,14} & c_{10,15} \\
\end{array}
      \right],
	$$ where the coefficients $c_{i,j}$, $i\leq 10$ and $11 \leq j\leq 15$, are elements from a finite field $GF(q)$. Since we do not show the $I_{k}$ matrix, the index $j$ for the non-zero coefficients is in the range between $k+1$ and $n$.
Note that the number of non-zero elements in $P_1$ per row is $w=3$ and per column is $l=6$. A transposed $P_1$ is graphically represented in Figure \ref{local}c.
The non-zero elements are represented with the shaded blocks in Figure \ref{local}c.
The average repair bandwidth for a single failure is 6 and for double failures is 9.
For comparative purposes we graphically represent the parity parts of an [14, 10] RS and an [16, 10] LRC in Figure \ref{local}a and \ref{local}b, respectively.
%The average repair bandwidth depends on the locality.
%The locality for different codes is graphically represented in Figure \ref{local}.
As we can see the RS code has the biggest locality $l=10$.
Consequently, a transfer of 10 blocks is required to repair any systematic or parity block when the RS code is used.
The [16, 10] LRC has locality equal to 5 for the local parity blocks and 10 for the global parity blocks.
Therefore, it requires a transfer of 5 blocks to repair a single failure of the systematic and the local parity blocks, while it takes 10 blocks to repair the global parities.
Hence, the average repair bandwidth for a single failure with the [16, 10] Azure LRC is $(5\times12+10\times4)/16=6.25$.
However, the [16, 10] Xorbas LRC reduces the number of transferred blocks for a repair of any single global parity block to 5 by introducing an implied parity block.
Thus, it has a lower average repair bandwidth compared to the Azure LRC implementation.
%Hence, the average repair bandwidth for a single failure is $(5\times12+10\times4)/16=6.25$.
%While the average repair bandwidth for both data and parity blocks is the same with the balanced code.
A comparison of the performance metrics for the $[15,10]$ code with parity part $P_1$ with 3-replication, the [14, 10] RS and the [16, 10] Xorbas LRC is presented in Table \ref{compare}, while additionally the Azure LRC is added in Figure \ref{repair}.
The way how we calculate the MTTDL is described in Section \ref{reliability}.

The next example of an [16, 10] code for $l=5$ and $w=3$ shows even a better performance when double failures occur.
The average repair bandwidth for a single failure is 5, while it is 7 for double failures.
This code tolerates up to any 3 failures and recovers the data successfully with 99.45\%, 96.02\% and 79.66\% from 4, 5 and 6 failures, respectively.
Its parity part is given as
%\vspace{0.1cm}
$$
P_2=\left[\small
      \begin{array}{c@{\hspace{0.5em}}c@{\hspace{0.5em}}c@{\hspace{0.5em}}c@{\hspace{0.5em}}c@{\hspace{0.5em}}c@{\hspace{0.5em}}c@{\hspace{0.5em}}c@{\hspace{0.5em}}c@{\hspace{0.5em}}c@{\hspace{0.5em}}c@{\hspace{0.5em}}c@{\hspace{0.5em}}c}
c_{1,11} & 0 & 0 & 0 & c_{1,15} & c_{1,16} \\
0 & c_{2,12} &  c_{2,13} & 0 & 0 &  c_{2,16} \\
0 & c_{3,12} & 0 & 0 & c_{3,15} & c_{3,16} \\
0 & 0 & c_{4,13} & c_{4,14} & c_{4,15} & 0 \\
c_{5,11} & 0 & 0 & c_{5,14} & 0 & c_{5,16} \\
0 & c_{6,12} & c_{6,13} & c_{6,14} & 0 & 0 \\
c_{7,11} & c_{7,12} & 0 & c_{7,14} & 0 & 0 \\
c_{8,11} & 0 & 0 & c_{8,14} & c_{8,15} & 0 \\
0 & 0 & c_{9,13} & 0 & c_{9,15} & c_{9,16} \\
c_{10,11} & c_{10,12} & c_{10,13} & 0 & 0 & 0 \\
\end{array}
      \right].
$$
%\vspace{0.2cm}
We depict the average repair bandwidth for different number of block failures for the RS, Xorbas LRC, Azure LRC and BLRCs in Figure \ref{repair}, under the condition of almost equal MTTDL.
BLRCs achieve an average repair bandwidth that is less than or equal to the average repair bandwidth with the other codes in case of a single block failure, while it is always less than the average repair bandwidth achieved with the other codes in case of double block failures.
Note that the storage overhead is less with the BLRC compared to the Xorbas LRC when they achieve the same average repair bandwidth for a single failure.

\begin{figure}
\begin{minipage}[b]{0.5\linewidth}
\centering
\includegraphics[width=8.5cm,height=6cm]{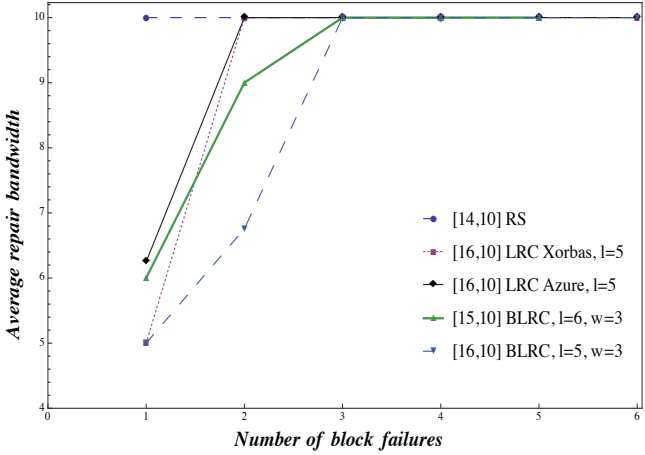}
\end{minipage}
\caption{Average repair bandwidth for different number of block failures}
\label{repair}
\end{figure}

The next example shows how the repair bandwidth can be reduced even more, but then the fault tolerance is worse. This has a direct impact on the reliability, i.e., MTTDL.
The parity part of an $[16, 10]$ code for $l=3$ or $l=4$ and $w=2$ is
$$
P_3=\left[\small
      \begin{array}{c@{\hspace{0.5em}}c@{\hspace{0.5em}}c@{\hspace{0.5em}}c@{\hspace{0.5em}}c@{\hspace{0.5em}}c@{\hspace{0.5em}}c@{\hspace{0.5em}}c@{\hspace{0.5em}}c@{\hspace{0.5em}}c@{\hspace{0.5em}}c@{\hspace{0.5em}}c@{\hspace{0.5em}}c}
c_{1,11} & 0 & 0 & 0 & 0 & c_{1,16} \\
0 & 0 & 0 & c_{2,14} & c_{2,15} & 0 \\
0 & 0 & 0 & c_{3,14} & 0 & c_{3,16} \\
0 & 0 & 0 & 0 & c_{4,15} & c_{4,16} \\
c_{5,11} & 0 & 0 & c_{5,14} & 0 & 0 \\
0 & c_{6,12} & 0 & 0 & c_{6,15} & 0 \\
c_{7,11} & c_{7,12} & 0 & 0 & 0 & 0 \\
c_{8,11} & 0 & c_{8,13} & 0 & 0 & 0 \\
0 & 0 & c_{9,13} & 0 & c_{9,15} & 0 \\
0 & c_{10,12} & c_{10,13} & 0 & 0 & 0 \\
\end{array}
      \right].
$$
\vspace{0.1cm}
When applying this code the average repair bandwidth for a single block failure is reduced to 3.33, while for double block failures to 5.22.
On the other hand, the fault tolerance is worse compared to the [16, 10] code for $l=5$ and $w=3$.
Thus, the MTTDL is reduced from $5.7378\times 10^{14}$ days to $7.2338\times 10^{8}$ days for an [16, 10] code when $l=5$, $w=3$ and $l=3.33$, $w=2$, respectively.

An overview of the performance metrics for the codes presented in this Section is given in Table \ref{compare}.

\section{Reliability Analysis} \label{reliability}
\setlength{\tabcolsep}{0.5em} % for the horizontal padding
{\renewcommand{\arraystretch}{1.4}
\begin{table*}[tb]
%\vspace{0.5cm}
\caption{Comparison summary of performance metrics for 3-replication, RS, Xorbas LRC and BLRCs}
\vspace{0.25cm}
\centering
\begin{tabular}{*{15}{|c|}}
\hline
Scheme & Storage overhead & \makecell{Avr. repair bandwidth \\ (single failure)} & \makecell{Avr. repair bandwidth \\ (double failure)} & MTTDL (days) & \makecell{Update \\ complexity}\\
\hline
3-replication & 2x  & 1x & 1x & $2.3079\times 10^{10}$ & 3\\ %[4ex]
\hline
[14, 10] RS code & 0.4x  & 10x & 10x & $3.3118\times 10^{13}$ & 5\\ %[4ex]
\hline
[16, 10] Xorbas LRC, $l=5$ & 0.6x  & 5x & 10x & $1.2180\times 10^{15}$ & 6\\ %[4ex]
\hline
%LRC (16, 10, 5) Azure & 0.6x  & 6.25x & 10x &  & 6\\
%\hline
[15, 10] BLRC, $l=6,w=3$ & 0.5x  & 6x & 9x & $3.3647\times 10^{14}$ & 4\\ %[4ex]
\hline
[16, 10] BLRC, $l=5,w=3$ & 0.6x  & 5x & 7x & $5.7378\times 10^{14}$ & 4\\ %[4ex]
\hline
[16, 10] BLRC, $l=3$ or $4, w=2$ & 0.6x  & 3.33x & 5.22x & $7.2338\times 10^{8}$ & 3\\ %[4ex]
\hline
\end{tabular}
\label{compare}
\end{table*}
}
%\vspace{0.5cm}
%the data blocks of each large le are grouped in stripes of 10 and for each such set, 4 parity blocks are created.

We perform a reliability analysis by calculating the MTTDL with a Markov model.
The authors in \cite{journals/pvldb/SathiamoorthyAPDVCB13} report values from the Facebook cluster and show that the $[16, 10]$ Xorbas LRC provides significantly longer MTTDL compared to the $[14, 10]$ RS and 3-replication.

In our analysis, we use the same parameters as in \cite{journals/pvldb/SathiamoorthyAPDVCB13} in order to compare the results.  %so that the analysis is meaningful and comparable with other schemes.
The total size of the cluster data is $C=30PB$ and this data is stored in $N=3000$ nodes.
The mean time to failure of a disk node is 4 years (=1/$\lambda$) and the block size is $B=256MB$.
The node failures are independent.
The bandwidth for cross-rack communication for repairs is limited to $\gamma=1Gbps$.
Under an $[15, 10]$ code, each stripe consists of 15 blocks where each block is placed in different racks to provide a higher fault tolerance.
Thus the total number of stripes in the system is $C/(n B)$ where $n=15$.
The MTTDL of a stripe is calculated by using the Markov model shown in Figure \ref{markov}.
Each state in the Markov model represents the number of available (non-failed) blocks (data and parity blocks).
The circles denote the states when the system is up and the squares denote the states when there is a data loss in the system, i.e., the system is down.

Let $\lambda$ denotes the failure rate of a single block.
The blocks are distributed in different nodes and the failure rate per node is $\lambda$.
When the state is $i$, i.e., there are $i$ available blocks in a stripe, the failure rate is $i\lambda$.
Consequently, the transition rate from State 15 to State 14 is $15 \lambda$.
Since BLRCs are not MDS codes, there are two possible transitions from State 12.
One of the transitions is to State 11 where there are 4 decodable failures and the other one is to State 11F which represents a state with 4 non-decodable failures.
The percentage of 4 decodable failures is $p_4 = 99.2674\%$.
Therefore the transition rate to State 11 is $12 \lambda p_4$ and to State 11F is $12 \lambda (1-p_4)$.
The same situation happens when transitioning from State 11 to State 10 and State 10F where $p_5=89.677\%$.
When 6 blocks are lost, i.e., only 9 blocks are available in State 9, the lost blocks from the stripe cannot be recovered.
That is why State 9 is shown as a down state.
The lost blocks from the stripe are also non-recoverable in the States 11F and 10F.

In the reverse direction $\rho_i$ denotes the repair rate.
The rate at which a block is repaired depends on the number of downloaded blocks (the locality), the block size and the bandwidth dedicated for repairs.
For instance, a repair of any single data or parity block requires downloading 6 blocks, i.e., $\rho_1=\frac{\gamma}{6 B}$.
Any two lost blocks are repaired by downloading 9 blocks, while a repair of more than 2 lost blocks requires a transfer of 10 blocks.
Thus, $\rho_2=\frac{\gamma}{9 B}$ and $\rho_3=\rho_4=\rho_5=\frac{\gamma}{10 B}$.
The MTTDL of the system is calculated as:
\begin{equation}
MTTDL=\frac{MTTDL_{stripe}}{C/(n B)}.
\label{MTTDL}
\end{equation}

The MTTDL values for 3-replication, the [14, 10] RS, the [16, 10] Xorbas LRC and few BLRCs are presented in Table \ref{compare}.
We observe that the fast repair and the high fault tolerance lead to a high reliability with the [15, 10] and [16, 10] BLRCs.

\begin{figure}
\begin{minipage}[b]{0.5\linewidth}
\centering
\includegraphics[width=3.4in]{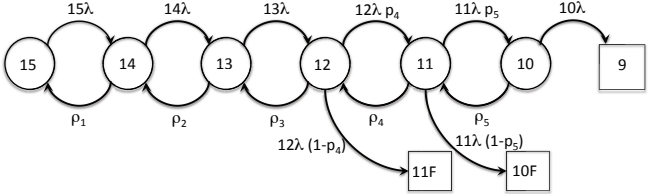}
\end{minipage}
\caption{Markov model for an [15, 10] code where circles represent the states when the data can be recovered and squares represent the states when the data is unrecoverable}
\label{markov}
\end{figure}

\section{Conclusions} \label{Conclusions}
We defined a new family of balanced locally repairable codes (BLRCs).
A novel property of the codes that we presented is that there is no strict requirement that the repair locality is a fixed small number $l$, and it may be either $l$ or $l+1$.
Advantageously, this provides the flexibility to construct BLRCs for arbitrary values of $n$ and $k$ which allows a general construction of LRCs.
%In the presented codes, the parity blocks depend in a balanced manner from the systematic data blocks and the locality dependence is either $l$ or $l+1$. The codes can be constructed for arbitrary values of $n$, $k$ and $l$ which solves the open problem given by Tamo et al. about the general construction of LRCs \cite{6620540}.
The properties of the presented codes are:
%DeErasure codes can have also low repair locality even when double failures occur.
%The focus of the present paper is on codes with locality.
low storage overhead, low average repair bandwidth for a single failure and double failures, high reliability and low update complexity.
%We used different metrics to compare the performance of BLRCs with existing LRCs.
%simplified design, implementation and verification;
%Data storage applications require codes with small storage overhead, low locality and large distance.

%\vspace{0.5cm}
\bibliographystyle{IEEEtran}
\bibliography{refer}

\end{document}